\documentclass[12pt,leqno,letterpaper]{article}

\usepackage{amsmath,amsthm,enumerate,amssymb}
\usepackage[utf8]{inputenc}

\usepackage[english]{babel}
\usepackage{graphicx}
\usepackage{color}

\usepackage{graphicx}

\newtheorem{theorem}{Theorem}[section]
\newtheorem{proposition}[theorem]{Proposition}
\newtheorem{lemma}[theorem]{Lemma}
\newtheorem{corollary}[theorem]{Corollary}

\newtheorem{remark}[theorem]{Remark}

\begin{document}

\title{Perspectives and completely positive maps}
\author{Frank Hansen}
\date{August 7, 2016}

\maketitle

\begin{abstract} We study the filtering of the perspective of a regular operator map of several variables through a completely positive linear map.  By this method we are able to extend known operator inequalities of two variables to several variables; with applications in the theory of operator means of several variables. We also extend Lieb-Ruskai's convexity theorem from two to $ n+1 $ operator variables. 
\\[1ex]
{\bf MSC2010} classification: 47A63\\[1ex]
{\bf{Key words and phrases:}}  partial traces of operator means; Lieb-Ruskai's convexity theorem for several variables.
\end{abstract}

\section{Introduction} We study the filtering of a regular operator map through a completely positive linear map $ \Phi. $ A main result is the inequality
\[
 F\bigl(\Phi(A_1),\dots,\Phi(A_k)\bigl)\le\Phi\bigl(F(A_1,\dots,A_k)\bigr),
\]
where $ A_1,\dots,A_k $ are positive definite operators on a Hilbert space of finite dimensions, and $ F $ is a convex regular operator map of $ k $ variables. If  $ G_k $ denotes any of the various geometric means of $ k $ variables studied in the literature we obtain as a special case the inequality
\[
\Phi\bigl(G_k(A_1,\dots,A_k)\bigr)\le  G_k\bigl(\Phi(A_1),\dots,\Phi(A_k)\bigl).
\]
This inequality extends result in the literature for $ k=2, $ for geometric means of $ k $ variables that may be obtained inductively by the power mean of two variables, and for means that are limits of such means, including the Karcher mean \cite{kn:bhatia:2012}. 

We extend Lieb-Ruskai's convexity theorem from two to $ n+1 $ operator variables. For $ n=2 $ we obtain in particular that the map
\[
L(A,B,C)= C^* B^{-1/2}\bigl(B^{1/2}A^{-1}B^{1/2}\bigr)^{1/2}B^{-1/2} C
\]
is convex in arbitrary $ C $ and positive definite $ A $ and $ B. $ In addition, 
\[
L\bigl(\Phi(A),\Phi(B),\Phi(C)\bigr)\le \Phi\bigl(L(A,B,C)\bigr)
\]
for any completely positive linear map $ \Phi $ between operators acting on finite dimensional Hilbert spaces. In particular, this includes quantum channels and partial traces. For commuting $ A $ and $ B $ the generalised Lieb-Ruskai map reduces to
\[
L(A,B,C)=C^*A^{-1/2}B^{-1/2}C.
\]
In particular, $ L(A,A,C)=C^*A^{-1}C. $

\section{Preliminaries}

Let $ \mathcal D\subseteq B(\mathcal H)\times\cdots\times B(\mathcal H) $ be a convex domain, where $ B(\mathcal H) $ is the algebra of bounded linear operators on a Hilbert space $ \mathcal H. $ 

We defined \cite[Definition 2.1]{kn:hansen:2014:4} the notion of a regular map 
$ F\colon\mathcal D\to B(\mathcal H), $  generalising the notion of a spectral function of Davis for functions of one variable, the notion of a regular matrix map of two variables by the author \cite{kn:hansen:1983}, and the notion of a regular operator map of two variables  \cite[Definition 2.1]{kn:hansen:2014:2} by Effros and the author. Loosely speaking, a regular map is unitarily invariant and reduces block matrices in a simple and natural way. It retains regularity when compressed to a subspace.

Although we often restrict the study to finite dimensional spaces it is convenient to consider only such regular maps that may be defined also on an infinite dimensional Hilbert space $ \mathcal H. $ Since $ \mathcal H $  in this case is isomorphic to $ \mathcal H\oplus\mathcal H $ this allows us to use block matrix techniques without imposing dimension conditions. Furthermore, it implies that a regular map is well-defined regardless of the underlying Hilbert space. We may thus port a regular map unambiguously from one Hilbert space to another.
In this paper we only consider domains of the form,
\[
\mathcal D^k(\mathcal H)=\{(A_1,\dots,A_k)\mid A_1,\dots,A_k\ge 0\},
\]
of $ k $-tuples of positive semi-definite operators, or domains,
\[
\mathcal D_+^k(\mathcal H) = \{ (A_1,\dots,A_k)\mid A_1,\dots,A_k>0\},
\]
of $ k $-tuples of positive definite and invertible operators acting on a Hilbert space $ \mathcal H. $ The latter is the natural type of domain for perspectives.

\subsection{Jensen's inequality for regular operator maps}

The following result was proved for $ \mathcal H=\mathcal K $ in \cite[Theorem 2.2 (i)]{kn:hansen:2014:4}. It is just an exercise to generalise the statement and obtain the following:

\begin{lemma}\label{lemma: Jensens operator inequality}
Let $ F\colon\mathcal D^k(\mathcal H)\to B(\mathcal H)_\text{sa} $ be a convex regular map, and take a contraction $ C\colon\mathcal H\to\mathcal K $ of $ \mathcal H $ into a Hilbert space $ \mathcal K. $ If $ F(0,\dots,0)\le 0 $ then the inequality
\[
F(C^*A_1C,\dots,C^*A_kC)\le C^* F(A_1,\dots,A_k)C
\]
holds for $ k $-tuples $ (A_1,\dots,A_k) $ in $ \mathcal D^k(\mathcal K). $
\end{lemma}

The next result reduces to \cite[Theorem 2.2 (ii)]{kn:hansen:2014:4} for $ \mathcal H=\mathcal K $ and $ n=2. $ Since the generalisation is quite straight-forward we leave the proof to the reader. 

\begin{theorem}[Jensen's inequality for regular operator maps]\label{Jensen's inequality for regular operator maps}\-\\
Let $ F\colon\mathcal D^k(\mathcal H)\to B(\mathcal H)_\text{sa} $ be a convex regular map
and let $ C_1,\dots,C_n\colon\mathcal H\to\mathcal K $ be mappings of $ \mathcal H $ into (possibly another) Hilbert space $ \mathcal K $ such that
\[
C_1^*C_1+\cdots+C_n^*C_n=1_{\mathcal H}.
\]
Then the inequality
\[
F\Bigl(\sum_{i=1}^n C_i^* A_{i1} C_i,\dots,\sum_{i=1}^n C_i^*A_{ik}C_i\Bigr)
\le\sum_{i=1}^n C_i^* F(A_{i1},\dots,A_{ik})C_i
\]
holds for $ k $-tuples $ (A_{i1},\dots,A_{ik}) $ in $ \mathcal D^k(\mathcal K) $ for $ i=1,\dots,n. $
\end{theorem}

\begin{corollary}\label{Jensen's inequality for completely positive unital maps}
Let $ \Phi\colon B(\mathcal H)\to B(\mathcal K) $ be a completely positive unital linear map between operators on Hilbert spaces of finite dimensions, and let
$ F $ be a convex regular map. Then
\[
F\bigl(\Phi(A_1),\dots,\Phi(A_k)\bigr)\le\Phi\bigl(F(A_1,\dots,A_k)\bigr)
\]
for $ (A_1,\dots,A_k)\in\mathcal D_k(\mathcal H). $
\end{corollary}

\begin{proof}
By Choi's decomposition theorem there exist operators $ C_1,\dots,C_n $ in $ B(\mathcal K,\mathcal H) $ with $ C_1^*C_1+\cdots+C_n^*C_n=1_{\mathcal K} $ such that
\[
\Phi(A)=\sum_{i=1}^n C_i^* A C_i\qquad\text{for}\qquad A\in B(\mathcal H).  
\]
The statement now follows by the preceding theorem by choosing
\[
(A_{i1},\dots,A_{ik})=(A_1,\dots,A_k)
\]
for $ i=1,\dots,n. $
\end{proof}

Davis \cite[Corollary]{kn:davis:1957} proved that $ f\bigl(\Phi(A)\bigr)\le \Phi\bigl(f(A)\bigr) $ for an operator convex function $ f $ with $ f(0)=0 $ and a completely positive linear map $ \Phi $ with $ \Phi(1)\le 1. $ Jensen's operator inequality is the slightly more general statement
\[
f\Bigl(\sum_{i=1}^n C_i^* A_i C_i\Bigr)\le\sum_{i=1}^n C_i^* f(A_i) C_i
\]
for tuples $ (A_1,\dots,A_n) $ and operators $ C_1,\dots,C_n $ with $ C_1^*C_1+\cdots+C_n^*C_n=1, $
see \cite[Theorem 2.1 (iii)]{kn:hansen:1982} and \cite{kn:hansen:2003:2}. Jensen's inequality for regular operator maps may in the same way be considered a generalisation of Corollary~\ref{Jensen's inequality for completely positive unital maps}.

\section{Perspectives}

We introduced the perspective \cite[Definition 3.1]{kn:hansen:2014:4} of a regular operator map of $ k $ variables as a generalisation of the operator perspective of a function of one variable defined by Effros \cite{kn:effros:2009:1}. A key result is that the perspective $ \mathcal P_F $ of a convex regular operator map $ F\colon\mathcal D_+^k(\mathcal H)\to B(\mathcal H) $ of $ k $ variables is a convex positively homogenous regular operator map of $ k+1 $ variables \cite[Theorem 3.2]{kn:hansen:2014:4}.

\begin{theorem}\label{theorem: filtering of perspective by a CP-map}
Let $ \Phi\colon B(\mathcal H)\to B(\mathcal K) $ be a completely positive linear map between operators on Hilbert spaces of finite dimensions, and let $ F\colon\mathcal D_+^k(\mathcal H)\to B(\mathcal H) $ be a convex regular map. Then
\[
\mathcal P_F\bigl(\Phi(A_1),\dots,\Phi(A_{k+1})\bigr)\le \Phi\bigl( \mathcal P_F(A_1,\dots,A_{k+1})\bigr),
\]
for operators $ (A_1,\dots,A_{k+1}) $ in $ \mathcal D^k_+(\mathcal H), $ where $ \mathcal P_F $ is the perspective of $ F.  $
\end{theorem}

\begin{proof} We extend an idea of Ando \cite{kn:ando:1979} from functions of one variable to regular operators maps.
To a fixed positive definite $ B\in B(\mathcal H) $ we set 
\[
\Psi(X)=\Phi(B)^{-1/2}\Phi(B^{1/2}XB^{1/2})\Phi(B)^{-1/2}
\]
and notice that $ \Psi\colon B(\mathcal H)\to B(\mathcal K) $ is a unital linear map. By the definition of complete positivity, we realise that also $ \Psi $ is completely positive. Since $ F $ is convex we may thus apply Corollary~\ref{Jensen's inequality for completely positive unital maps} and obtain
\[
\begin{array}{l}
F\bigl(\Psi(B^{-1/2}A_1B^{-1/2}),\dots,\Psi(B^{-1/2}A_kB^{-1/2})\bigr)\\[1.5ex]
\hskip 5em\le \Psi\bigl( F(B^{-1/2}A_1B^{-1/2},\dots,B^{-1/2}A_kB^{-1/2})\bigr).
\end{array}
\]
Inserting $ \Psi $ we obtain the inequality
\[
\begin{array}{l}
 F\bigl(\Phi(B)^{-1/2}\Phi (A_1)\Phi(B)^{-1/2},\dots, \Phi(B)^{-1/2}\Phi (A_k)\Phi(B)^{-1/2}\bigr)\\[1.5ex]
\le \Phi(B)^{-1/2}\Phi\bigl(B^{1/2}F(B^{-1/2}A_1B^{-1/2},\dots,B^{-1/2}A_kB^{-1/2})  B^{1/2}\bigr)\Phi(B)^{-1/2}.
\end{array}
\]
By multiplying from the left and from the right with $ \Phi(B)^{1/2} $ we obtain
\[
\begin{array}{l}
\mathcal P_F\bigl(\Phi(A_1),\dots,\Phi(A_k),\Phi(B)\bigr)=\\[1.5ex]
\Phi(B)^{1/2}  F\bigl(\Phi(B)^{-1/2}\Phi (A_1)\Phi(B)^{-1/2},\dots,\Phi(B)^{-1/2}\Phi (A_k)\Phi(B)^{-1/2}\bigr)\Phi(B)^{1/2}\\[1.5ex]
\le \Phi\bigl(B^{1/2}F(B^{-1/2}A_1B^{-1/2},\dots,B^{-1/2}A_kB^{-1/2})  B^{1/2}\bigr)\\[1.5ex]
=\Phi\bigl(\mathcal P_F(A_1,\dots,A_k,B)\bigr),
\end{array}
\]
which is the assertion.
\end{proof}

Notice that we do not require $ \Phi $ to be unital or trace preserving in the above theorem. 

\begin{theorem}\label{theorem: homogeneous and convex map under partial tracing}
Let $ \Phi\colon B(\mathcal H)\to B(\mathcal K) $ be a completely positive linear map between operators on Hilbert spaces of finite dimensions, and
let  $ F\colon\mathcal D_+^{k+1}(\mathcal H)\to B(\mathcal H) $ be a convex and positively homogeneous regular map. Then
\[
F\bigl(\Phi(A_1),\dots,\Phi(A_{k+1})\bigr)\le\Phi\bigl( F(A_1,\dots,A_{k+1})\bigr)
\]
for positive definite $ A_1,\dots,A_{k+1} \in B(\mathcal H). $ 
\end{theorem}

\begin{proof} We proved elsewhere \cite[Proposition 3.3]{kn:hansen:2014:4} that a convex and positively homogeneous regular mapping $ F $ of $ k+1 $ variables is the perspective of its restriction
\[
G(A_1,\dots,A_k)=F(A_1,\dots,A_k,1)
\]
to $ k $ variables. Since $ G\colon\mathcal D_+^k(\mathcal H)\to B(\mathcal H)  $ is convex and regular the assertion follows from Theorem~\ref{theorem: filtering of perspective by a CP-map}.
\end{proof}

\begin{remark}
A geometric mean $ G $ of several variables is an example of a concave positively homogeneous regular map. The inequality in Theorem~\ref{theorem: homogeneous and convex map under partial tracing} thus reduces to
\[
G\bigl(\Phi(A_1),\dots,\Phi(A_k)\bigr)\ge\Phi\bigl( G(A_1,\dots,A_k)\bigr).
\]
This result was proved \cite[Theorem 4.1]{kn:bhatia:2012} for all geometric means that may be obtained inductively by an application of the power mean of two variables. By a limiting argument this was then extended to the Karcher mean. However, there exist geometric means that cannot be obtained in this way, for example the means introduced in \cite[Section 4.2]{kn:hansen:2014:4}.
\end{remark}

\section{Lieb-Ruskai's convexity theorem}

Lieb and Ruskai~\cite[Theorem 1]{kn:lieb:1974} proved convexity of the map
\[
L(A,K)= K^*A^{-1} K
\]
in pairs $ (A,K) $ of bounded linear operators on a Hilbert space, where $ A $ is positive definite.
Subsequently, Ando gave a very elegant proof of this result \cite[Theorem 1]{kn:ando:1979}. If $ K $ is positive definite we may write
\[
KA^{-1}K=K^{1/2}\bigl(K^{-1/2}AK^{-1/2}\bigr)^{-1} K^{1/2}
\] 
as the perspective of the function $ t\to t^{-1}. $ Since this function is operator convex, we obtain convexity of the perspective $ L(A,K), $ if $ K $ is restricted to positive definite operators. This however is enough to obtain the general result. Indeed, the set of $ (K,A) $ where $ \|K\|<1 $ and $ A\ge 1 $ is convex, and the embedding
\begin{equation}\label{embedding into positive operators}
K\to
\begin{pmatrix}
A & K^*\\
K & A
\end{pmatrix}>0
\end{equation}
is affine into positive definite operators. It thus follows that
\[
\begin{array}{rl}
(K,A)&\to \begin{pmatrix}
A & K^*\\
K & A
\end{pmatrix}
\begin{pmatrix}
A & 0\\
0 & A
\end{pmatrix}^{-1}
\begin{pmatrix}
A & K^*\\
K & A
\end{pmatrix}\\[3ex]
&=\begin{pmatrix}
     A+K^*A^{-1}K & 2K^*\\
     2K & A+KA^{-1}K^*
     \end{pmatrix}
\end{array}
\]
is convex in the specified set. In particular, $ (K,A)\to K^*A^{-1}K $ is convex.

M.B. Ruskai kindly informed the author that Lieb and Ruskai obtained their much cited convexity result unaware that it was proved much earlier in another context by Kiefer \cite{kn:kiefer:1959}.

\begin{proposition}\label{Lieb-Ruskai's theorem and partial traces}
Let $ \Phi\colon B(\mathcal H)\to B(\mathcal K) $ be a completely positive linear map between operators on Hilbert spaces of finite dimensions. The inequality
\[
\Phi(K)^*\Phi(A)^{-1}\Phi(K)\le\Phi(K^*A^{-1}K)
\]
is valid for positive definite $ A $ and arbitrary $ K. $
\end{proposition}

\begin{proof}  If we restrict $ K $ to positive definite operators the first inequality is already contained in Theorem~\ref{theorem: filtering of perspective by a CP-map}. The same block matrix construction as in (\ref{embedding into positive operators}) applied to the completely positive linear map $ \Phi\otimes 1_2 $ then leads to the inequality
\[
\Phi(A)+\Phi(K^*)\Phi(A)^{-1}\Phi(K)\le\Phi(A+K^*A^{-1}K)
\]
for $ A\ge 1 $ and $ \|K\|< 1, $ and the statement follows.
\end{proof}

Notice that the above inequality was obtained in \cite[Corollary 3.1]{kn:ando:1979} if $ K $ is positive definite; cf. also \cite[Theorem 2 and Theorem 3]{kn:lieb:1974}.

There is another way to consider Lieb-Ruskai's convexity theorem which points to generalisations of the result to more than two operators. The geometric mean $ G_1 $ of one positive definite operator is trivially given by $ G_1(A)=A. $ It is a concave regular map and its inverse
\[
A\to G_1(A)^{-1}=A^{-1}
\]
is thus a convex regular map. The perspective
\[
\mathcal P_{G_1^{-1}}(A,B)=B^{1/2}G_1(B^{-1/2}AB^{-1/2})^{-1}B^{1/2}=BA^{-1}B=L(A,B)
\]
is therefore a convex regular map by \cite[Theorem 3.2]{kn:hansen:2014:4}, and it is increasing when filtered through a completely positive linear map by Theorem~\ref{theorem: filtering of perspective by a CP-map}. A similar construction may be carried out for any number of operator variables.

\begin{theorem}\label{theorem: convexity of L(A_1,...,A_n,C)}
Let $ G_n $ be an extension of the function
\[
(t_1,\dots,t_n)\to t_1^{1/n}\cdots t_n^{1/n}\qquad t_1,\dots,t_n>0
\]
to an operator map defined in positive definite invertible operators on a Hilbert space $ \mathcal H. $ Furthermore, suppose that  $ G_n $
is a positively homogeneous regular operator map which is concave, self-dual and congruence invariant, cf. the discussions in \cite{kn:ando:2004:1} and \cite{kn:hansen:2014:4}.
The operator map
\[
L(A_1,\dots,A_n,C)= C G_n(A_1,\dots,A_n)^{-1} C,
\]
is then convex in positive definite and invertible operators.
\end{theorem}

\begin{proof}
The (geometric) mean $ G_n $ is a positive, concave, and regular map. The inverse
\[
G_n(A_1,\cdots,A_n)^{-1}=G_n(A^{-1},\dots,A_n^{-1})
\]
is therefore convex and regular.  The perspective
\[
\begin{array}{l}
\mathcal P_{G_n^{-1}}(A_1,\dots,A_n,C)\\[1.5ex]
=C^{1/2} G_n\bigl(C^{-1/2}A_1C^{-1/2},\dots, C^{-1/2}A_n C^{-1/2}\bigr)^{-1} C^{1/2}\\[1.5ex]
=C^{1/2} G_n\bigl(C^{1/2}A_1^{-1}C^{1/2},\dots, C^{1/2}A_n^{-1}C^{1/2}\bigr) C^{1/2}\\[1.5ex]
=C G_n\bigl(A_1^{-1},\dots, A_n^{-1}\bigr) C= C G_n(A_1,\dots,A_n)^{-1} C\\[1.5ex]
=L(A_1,\dots,A_n,C),
\end{array}
\]
where we used self-duality and congruence invariance of the geometric mean. It now follows, by \cite[Theorem 3.2]{kn:hansen:2014:4}, that $ L $ is a convex regular map. 
\end{proof}

\begin{remark}
It is interesting to notice that Theorem~\ref{theorem: convexity of L(A_1,...,A_n,C)} alternatively may be obtained by adapting the arguments of Ando in \cite[Theorem 1]{kn:ando:1979}, and that this way of reasoning even imparts convexity of the map
\[
L(A,B,C)=C^* G_2(A,B)^{-1} C,
\]
where $ C $ now is arbitrary and $ A,B $ are positive definite and invertible. The argument uses the well-known fact that a block matrix of the form
\[
\begin{pmatrix}
A & C\\
C^* & B
\end{pmatrix},
\]
where $ A $ is positive definite and invertible, is positive semi-definite if and only if $ B\ge C^*A^{-1}C. $ 
Indeed, by taking   $ \lambda\in[0,1] $ and setting
\[
\begin{array}{rl}
C&=\lambda C_1+(1-\lambda)C_2\\[1ex]
T&= \lambda C_1^* G_2(A_1,B_1)^{-1}C_1+(1-\lambda)C_2^*G_2(A_2,B_2)^{-1}C_2
\end{array}
\]
we obtain the equality
\[
\begin{array}{l}
X=\begin{pmatrix}
\lambda G_2(A_1,B_1)+(1-\lambda)G_2(A_2,B_2) & C\\
C^* & T
\end{pmatrix}\\[4ex]
=\lambda\begin{pmatrix}
                G_2(A_1,B_1) & C_1\\
                C_1^*                 & C_1^* G_2(A_1,B_1)^{-1}C_1
                \end{pmatrix}\\[3ex]
\hskip 10em+(1-\lambda)\begin{pmatrix}
                G_2(A_2,B_2) & C_2\\
                C_2^*                 & C_2^* G_2(A_2,B_2)^{-1}C_2
                \end{pmatrix}.
\end{array}
\]
Since the two last block matrices by construction are positive semi-definite, we obtain that the block matrix $ X $ is positive semi-definite. Therefore,
\[
T\ge C^*\bigl(\lambda G_2(A_1,B_1)+(1-\lambda)G_2(A_2,B_2)\bigr)^{-1} C.
\]
We thus obtain
\[
\begin{array}{l}
\lambda L(A_1,B_1,C_1)+(1-\lambda) L(A_2,B_2,C_2)\\[1.5ex]
=\lambda C_1^* G_2(A_1,B_1)^{-1}C_1+(1-\lambda)C_2^*G_2(A_2,B_2)^{-1}C_2=T\\[1.5ex]
\ge C^*\bigl(\lambda G_2(A_1,B_1)+(1-\lambda)G_2(A_2,B_2)\bigr)^{-1} C\\[1.5ex]
\ge C^* G_2(\lambda A_1+(1-\lambda) A_2, \lambda B_1+(1-\lambda) B_2)^{-1}C\\[1.5ex]
=L(\lambda A_1+(1-\lambda)A_2, \lambda B_1+(1-\lambda B_2), \lambda C_1+(1-\lambda C_2),
\end{array}
\]
where we in the last inequality used concavity of the geometric mean and operator convexity of the inverse function.
\end{remark}

It seems mysterious that we in the last proof only used concavity of $ G_2 $ while we in Theorem~\ref{theorem: convexity of L(A_1,...,A_n,C)} used self-duality and congruence invariance in addition. However, if we want $ L(A,B,C) $ to be positively homogeneous, then $ G_2 $ must have the same property; and if we also want $ G_2 $ to be an extension of the geometric mean of positive numbers, then the geometric mean of two operators is the only solution satisfying all these requirements, cf. \cite[Proposition 3.3]{kn:hansen:2014:4}.
This way of reasoning extends to any number of variables and we obtain:

\begin{corollary} Let $ G_n $ be any geometric mean of $ n $ positive semi-definite and invertible operators. The operator function
\begin{equation}\label{convexity of L in the general case}
L(A_1,\dots,A_n,C)=C^* G_n(A_1,\dots,A_n)^{-1} C
\end{equation}
is convex in arbitrary $ C $ and positive definite and invertible $ A_1,\dots,A_n $ acting on a Hilbert space.
\end{corollary}

\begin{corollary}
Let $ \Phi\colon B(\mathcal H)\to B(\mathcal K) $ be a completely positive linear map between operators on Hilbert spaces of finite dimensions. The inequality
\[
L\bigl(\Phi(C),\Phi(A_1),\dots,\Phi(A_n)\bigr)\le\Phi\bigl(L(C,A_1,\dots,A_n)\bigr)
\]
is valid for positive definite $ A_1,\dots,A_n $ and $ C. $
\end{corollary}

It is known that the geometric mean of two variables is the unique extension of the function $ (t,s)\to t^{1/2} s^{1/2} $ to a positively homogeneous, regular and concave operator map \cite{kn:hansen:2014:2}. Therefore,
\[
L(A,B,C)= C B^{-1/2}\bigl(B^{1/2}A^{-1}B^{1/2}\bigr)^{1/2}B^{-1/2} C
\]
is the only sensible extension of Lieb-Ruskai's map to three positive definite and invertible operators with symmetry condition $ L(A,B,C)=L(B,A,B). $ Without the symmetry condition there are other solutions. The weighted geometric mean,
\[
G_2(\alpha; A,B)=B^{1/2}\bigl(B^{-1/2}AB^{-1/2}\bigr)^\alpha B^{1/2}\qquad 0\le\alpha\le 1,
\]
is the perspective of the operator concave function $ t\to t^\alpha $ and is therefore concave and congruent invariant \cite{kn:hansen:1983,kn:hansen:2014:2}. It is also manifestly self-dual. We can therefore apply a proof similar to the one used in the preceding theorem and obtain that the map
\[
L(\alpha; A,B,C)=C B^{-1/2}\bigl(B^{1/2}A^{-1}B^{/12}\bigr)^\alpha B^{-1/2} C
\]
is convex in positive semi-definite and invertible operators. Furthermore, it is positively homogeneous and therefore increasing when filtered through a completely positive linear map between operators on finite dimensional Hilbert spaces. It reduces to
\[
L(\alpha; A,B,C)=CA^{-\alpha} B^{-(1-\alpha)} C
\]
for commuting $ A $ and $ B. $

It is known that there for $ n\ge 3 $ exist many different extensions of the real function $ (t_1,\dots,t_n)\to t_1^{1/n}\cdots t_n^{1/n} $ to an operator mapping $ G_n $ satisfying the conditions in the preceding theorem, cf. \cite{kn:hansen:2014:4}.
Notice that if $ A_1,\dots,A_n $ commute then
\[
L(A_1,\dots,A_n,C)=C^* A_1^{-1/n}\cdots A_n^{-1/n}C
\]
and in particular $ L(A,\dots,A,C)=C^*A^{-1}C. $\\[1ex] 
{\small
{\bf Acknowledgments}
This work was supported by the Japanese Grant-in-Aid for scientific research 26400104.


\begin{thebibliography}{10}

\bibitem{kn:ando:1979}
T.~Ando.
\newblock Concavity of certain maps of positive definite matrices and
  applications to \uppercase{H}adamard products.
\newblock {\em Linear Algebra Appl.}, 26:203--241, 1979.

\bibitem{kn:ando:2004:1}
T.~Ando, C.-K. Li, and R.~Mathias.
\newblock Geometric means.
\newblock {\em Linear Algebra Appl.}, 385:305--334, 2004.

\bibitem{kn:bhatia:2012}
R.~Bhatia and R.L. Karandikar.
\newblock Monotonicity of the matrix geometric mean.
\newblock {\em Math. Ann.}, 353:1453--1567, 2012.

\bibitem{kn:davis:1957}
C.~Davis.
\newblock A \uppercase{S}chwarz inequality for convex operator functions.
\newblock {\em Proc. Amer. Math. Soc.}, 8:42--44, 1957.

\bibitem{kn:hansen:2014:2}
E.~Effros and F.~Hansen.
\newblock Non-commutative perspectives.
\newblock {\em Annals of Functional Analysis}, 5(2):74--79, 2014.

\bibitem{kn:effros:2009:1}
E.G. Effros.
\newblock A matrix convexity approach to some celebrated quantum inequalities.
\newblock {\em Proc. Natl. Acad. Sci. USA}, 106:1006--1008, 2009.

\bibitem{kn:hansen:1983}
F.~Hansen.
\newblock Means and concave products of positive semi-definite matrices.
\newblock {\em Math. Ann.}, 264:119--128, 1983.

\bibitem{kn:hansen:2003:2}
F.~Hansen and Pedersen G.K.
\newblock Jensen's operator inequality.
\newblock {\em Bull. London Math. Soc.}, 35:553--564, 2003.

\bibitem{kn:hansen:1982}
F.~Hansen and G.K. Pedersen.
\newblock Jensen's inequality for operators and \uppercase{L}{\"o}wner's
  theorem.
\newblock {\em Math. Ann.}, 258:229--241, 1982.

\bibitem{kn:hansen:2014:4}
Frank Hansen.
\newblock Regular operator mappings and multivariate geometric means.
\newblock {\em Linear Algebra and Its Applications}, 461:123--138, 2014.

\bibitem{kn:kiefer:1959}
J.~Kiefer.
\newblock Optimum experimental designs.
\newblock {\em J. Roy. Statist. Soc. Ser. B}, 21:272--310, 1959.

\bibitem{kn:lieb:1974}
E.H. Lieb and M.B. Ruskai.
\newblock Some operator inequalities of the \uppercase{S}chwarz type.
\newblock {\em Adv. in Math.}, 12:269--273, 1974.

\end{thebibliography}

\vfill

\noindent Frank Hansen: Institute for Excellence in Higher Education, Tohoku University, Sendai, Japan.\\
Email: frank.hansen@m.tohoku.ac.jp.
      }

\end{document}